\documentclass[conference,a4paper]{IEEEtran}
\IEEEoverridecommandlockouts

\usepackage{amsmath, amssymb, cite}
\usepackage{amsthm,bm,bbm}
\usepackage{mathtools}
\usepackage{amsthm,multirow,color,amsfonts}
\usepackage{tabulary}
\usepackage{subfigure}
\usepackage{graphicx}
\usepackage{setspace}
\usepackage{enumerate}
\usepackage{algorithm}
\usepackage{algpseudocode}
\usepackage{comment}
\usepackage{pdfpages}
\usepackage{hyperref}

\DeclareMathOperator*{\qq}{{\small q}}

\newcommand\independent{\protect\mathpalette{\protect\independenT}{\perp}}
\def\independenT#1#2{\mathrel{\rlap{$#1#2$}\mkern2mu{#1#2}}}

\usepackage[dvipsnames]{xcolor}

\newtheorem{prop}{Proposition}

\newtheorem{cor}{Corollary}

\newcommand{\RSWsum}{R_{\rm SW}^{\Sigma}}
\newcommand{\RKMsum}{R_{\rm KM}^{\Sigma}}
\newcommand{\RSsum}{R_{\rm S}^{\Sigma}}
\newcommand{\RSVsum}{R_{\rm SV}^{\Sigma}}
\newcommand{\RKMORsum}{R_{\rm KM-OR}^{\Sigma}}

\def\independenT#1#2{\mathrel{\rlap{$#1#2$}\mkern2mu{#1#2}}}

\usepackage{amssymb, cite}
\usepackage{bm,bbm}
\usepackage{mathtools}
\usepackage{multirow,color,amsfonts}
\usepackage{tabulary}
\usepackage{subfigure}
\usepackage{graphicx}
\usepackage{setspace}
\usepackage{enumerate}
\usepackage{comment}

\usepackage[utf8]{inputenc} 
\usepackage[T1]{fontenc}
\usepackage{url}
\usepackage{ifthen}
\usepackage{cite}

\usepackage{amsthm}
\usepackage{amsmath}

\begin{document}
\title{Distributed Structured Matrix Multiplication}

\author{
  \IEEEauthorblockN{Derya~Malak}
  \IEEEauthorblockA{Communication Systems Department,   EURECOM, 
                Biot Sophia Antipolis, FRANCE\\
                derya.malak@eurecom.fr}
\thanks{
Co-funded by the European Union (ERC, SENSIBILITÉ, 101077361). Views and opinions expressed are however those of the author only and do not necessarily reflect those of the European Union or the European Research Council. Neither the European Union nor the granting authority can be held responsible for them. 

This research was partially supported by a Huawei France-funded Chair towards Future Wireless Networks, and supported by the program ``PEPR Networks of the Future" of France 2030. 
}  
}

\maketitle

\begin{abstract}
We devise achievable encoding schemes for distributed source compression for computing inner products, symmetric matrix products, and more generally, square matrix products, which are a class of nonlinear transformations. To that end, our approach relies on devising nonlinear mappings of distributed sources, which are then followed by the structured linear encoding scheme, introduced by K\"orner and Marton. 
For different computation scenarios, we contrast our findings on the achievable sum rate with the state of the art to demonstrate the possible savings in compression rate. 
When the sources have special correlation structures, it is possible to achieve unbounded gains, as demonstrated by the analysis and numerical simulations.
\end{abstract}


\section{Introduction}
\label{sec:intro}
The inner product operation between two vectors captures the similarity between vectors and allows us to describe the lengths, angles, projections, vector norms, matrix norms induced by vector norms, orthogonality of vectors, polynomials, and a variety of other functions as well \cite{strang2022introduction}. 
Inner products are widely used in geometry and trigonometry using linear algebra and in applications spanning physics, engineering, and mathematics, e.g., to determine the convolution of functions \cite{strang2022introduction}, and the Fourier transform approximations, machine learning \cite{mousavi2019private} and pattern recognition \cite{bouscatie2023pattern}, e.g., the linear regression and the least squares models \cite{strang2022introduction}, and quantum computing, e.g., to describe the overlap between the two quantum states \cite{buhrman2001quantum}. 

With the advent of 
massive parallelization techniques and coded computation frameworks, modern distributed computing systems, e.g., MapReduce, 
Hadoop, and Spark, have been devised to implement the computationally intensive task of distributed matrix multiplication with low communication and computation cost \cite{tandon2017gradient}. To that end, novel coded matrix-multiplication constructions, e.g., polynomial codes \cite{yu2017polynomial}, \cite{dutta2019optimal}, gradient coding \cite{tandon2017gradient}, and Lagrange coded computing \cite{soleymani2021analog,YuRavSoAve2018}, have been designed to mitigate the costs, faulty nodes, and stragglers. The inner product computation serves as the primary building block of such settings. 

In this paper, we devise structured encoding schemes for distributed computing of inner products and symmetric and, more generally, square matrices via distributed matrix products. 
Our main contributions are summarized as follows:
\begin{itemize}
\item We devise a distributed encoding scheme that performs structured coding on nonlinear mappings of two distributed sources ${\bf A}$ and ${\bf B}$ to compute their inner product. 
\item We showcase the conditions for which the sum rate achieved by this structured coding is strictly less than the sum rate of distributed unstructured encoding of the sources ${\bf A}$ and ${\bf B}$ \cite{SlepWolf1973}. Here, the performance criterion is that both ${\bf A}$ and ${\bf B}$ cannot be decoded by the receiver. 
\item We derive achievable rate regions for distributed computation of symmetric and square matrices, a class of nonlinear transformations, with a vanishing probability of error, and determine example scenarios --- detailed in Corollaries~\ref{cor:innerproduct_length_m_binary}-\ref{cor:general_matrix_q3} --- with special correlation structures to guarantee, via the structured coding scheme of K\"orner and Marton \cite{korner1979encode}, significant savings over \cite{SlepWolf1973}.
\item We contrast the achievable rates with the existing approaches (e.g., \cite{SlepWolf1973,krithivasan2011distributed,pradhan2020algebraic}) via numerical examples.
\end{itemize}

{\bf Connections to the state of the art.} Slepian and Wolf have provided an unstructured coding technique for the asymptotic lossless compression of distributed source variables $X_1$ and $X_2$ at the minimum rate needed, i.e., $H(X_1,X_2)$ \cite{SlepWolf1973}. Han-Kobayashi \cite{han1987dichotomy} have provided a characterization to determine whether computing a general bivariate function $f(X_1,X_2)$ of two random sequences $\{X_{1i}\}$ and $\{X_{2i}\}$ from two correlated memoryless sources requires a smaller rate than $H(X_1,X_2)$. For distributed coding of a finite alphabet source $X$ with side information $Y$, Orlitsky and Roche have devised an unstructured coding scheme to achieve the minimum rate at which source $X$ has to compress for distributed computing of $f(X,Y)$ with vanishing error \cite{OR01}, exploiting K\"orner's characteristic graph $G_X$ and its entropy \cite{Kor73}. This scheme is equivalent to performing Slepian-Wolf encoding on the colors of the sufficiently large OR powers of $G_X$ given $Y$  \cite{feizi2014network,malak2022fractional,salehi2023achievable}.

K\"orner and Marton have devised a structured encoding strategy that minimizes the sum rate for distributed computing the modulo-two sum of doubly symmetric binary source (DSBS) sequences with a low probability of error \cite{korner1979encode}. Ahlswede and Han have tightened the rate region for general binary sources   \cite{ahlswede1983source} that embed the regions of \cite{korner1979encode} and \cite{SlepWolf1973}. The rate region for this problem has been extended to a larger class of source distributions \cite{nair2020optimal}, and for reconstructing the modulo-$q$ sum of the two sources in a $q$-ary prime finite field $\mathbb{F}_q$ at a sum rate of $2H(X_1 \oplus_q X_2)$ \cite{han1987dichotomy}. For computing a nonlinear function, the embedding of the function in a sufficiently large prime $\mathbb{F}_q$ \cite{krithivasan2011distributed}, \cite{pradhan2020algebraic}, and finding an injective mapping between the function and $X_1 \oplus_q X_2$, known as {\emph{structured binning}}, may provide savings over \cite{SlepWolf1973}. 

In this paper, leveraging these fundamental principles, we demonstrate further savings in compression for computing inner products as well as matrix products through devising nonlinear mappings of the sources followed by the linear encoding scheme in \cite{korner1979encode}, while requiring a smaller prime field size versus \cite{krithivasan2011distributed}, \cite{pradhan2020algebraic}.

{\bf Notation.} We denote by $H(X)$ the Shannon entropy of a discrete random variable $X$, which is drawn from probability mass function (PMF) $P_{X}$. Similarly, $H(\cdot,\cdot)$ and $H(\cdot\,\vert\, \cdot)$ denote the joint and conditional entropies, given a joint PMF $P_{X_1,X_2}$. Let $h(\epsilon)$ denote the binary entropy function for Bernoulli distributed $X$ with parameter $\epsilon\in [0,1]$, i.e., $X\sim {\rm Bern}(\epsilon)$. 

We denote by $X^n$ the length $n$ realization of $X$. The boldface notation ${\bf X}$ denotes a random matrix with elements in $\mathbb{F}_q$ and ${\bf X}^{\intercal}$ is its transpose. ${\bm 1}_m$ and ${\bm 0}_m$ are length $m$ all-ones and all-zeros vectors, respectively. $\mathbb{P}(A)$ is the probability of an event $A$, and  $1_{x\in A}$ is the indicator function which takes the value $1$ if $x\in A$, and $0$ otherwise. The inner product of two vectors in the vector space $V$ over a field $F$ is a scalar, which is a map $\langle \cdot ,\cdot \rangle :V\times V\to F$.

\section{System Model and Main Results}
\label{sec:mainresults}
We consider a distributed coding scenario with two sources and a receiver. The two sources are assigned matrix variables ${\bf A}\in\mathbb{F}_q^{m\times l}$ and ${\bf B}\in\mathbb{F}_q^{m\times l}$, respectively, that model two statistically dependent independent and identically distributed (i.i.d.) finite alphabet source matrix 
sequences with entries from a field whose characteristic is $q\geq 2$. The objective of the receiver is to compute a function $f({\bf A},{\bf B})$, which is the matrix product of ${\bf A}$ and ${\bf B}$, i.e., ${\bf A}^{\intercal}{\bf B} :\mathbb{F}_q^{m\times l}\times \mathbb{F}_q^{m\times l}\to \mathbb{F}_q^{l\times l}$. 
To that end, we devise a distributed encoding scheme that performs structured coding on nonlinear source mappings. 

We demonstrate different regimes where the sum rate achieved by {\emph{structured coding}} is strictly less than $H({\bf A},{\bf B})$ \cite{SlepWolf1973} so that the receiver can recover ${\bf A}^{\intercal}{\bf B}$, 
while it is desired that sources ${\bf A}$ and ${\bf B}$ cannot be decoded by the receiver, implying a security constraint. 
Exploiting \cite{korner1979encode}, we first study the special scenario of distributed computing of inner products, where the key idea is to implement linear structured coding of \cite{korner1979encode} on vector-wise embeddings of the sources   (Propositions~\ref{prop:KW_sum_rate_for_inner_product} and~\ref{prop:KW_sum_rate_for_inner_product_linear_source_mappings}), and propose a hybrid scheme that combines \cite{korner1979encode} with the {\emph{unstructured coding}} technique of \cite{OR01}. We contrast the achievable rate savings over unstructured coding \cite{SlepWolf1973}, and over linear embedding of ${\bf A}^{\intercal}{\bf B}$ in a prime finite field \cite{krithivasan2011distributed} and \cite{pradhan2020algebraic}. We generalize our scheme for computing symmetric matrices via distributed multiplication of the sources without being able to recover the sources in their entirety (Proposition~\ref{prop:KW_sum_rate_for_symmetric_matrix_product}) and provide a necessary condition for not being able to recover ${\bf A}$ and ${\bf B}$ (Proposition~\ref{prop:KW_sum_rate_for_inner_product+vs_SW_sum_rate}). Finally, we consider the distributed computation of square matrices (Proposition~\ref{prop:KW_sum_rate_for_general_matrix_product} and Corollary~\ref{cor:general_matrix_q3}). 
We showcase the compression performance of our schemes versus the existing results, demonstrating the achievable savings.

We next detail our main achievability results.

\subsection{Distributed Computation of Inner Products of Sources}
\label{sec:achievability_results_inner_product}
The distributed sources hold the even-length vector variables ${\bf A}=\begin{bmatrix}a_1 & a_2 & \hdots & a_m\end{bmatrix}^{\intercal}$ and ${\bf B}=\begin{bmatrix}b_1 & b_2 & \hdots & b_m\end{bmatrix}^{\intercal}$, respectively, with entries chosen from a prime field $\mathbb{F}_q$. The receiver aims to compute the inner product $f({\bf A},{\bf B})=\langle \,{\bf A} ,{\bf B} \,\rangle$.

We next present an achievable coding scheme for distributed computation of $\langle \,{\bf A} ,{\bf B} \,\rangle=\sum\nolimits_{i=1}^{m} a_i b_i$. The coding scheme relies on devising nonlinear mappings from each source and using the linear encoding scheme of K\"orner-Marton in \cite{korner1979encode}.

\begin{prop}
\label{prop:KW_sum_rate_for_inner_product}
{\bf (Distributed inner product computation.)}
Given two sequences of random vectors ${\bf A}=\begin{bmatrix}{\bf A}_1^{\intercal}&{\bf A}_2^{\intercal}\end{bmatrix}^{\intercal}\in\mathbb{F}_q^{m\times 1}$ and ${\bf B}=\begin{bmatrix}{\bf B}_1^{\intercal}&{\bf B}_2^{\intercal}\end{bmatrix}^{\intercal}\in\mathbb{F}_q^{m\times 1}$ of even length $m$, generated by two correlated memoryless $q$-ary sources, where ${\bf A}_1,{\bf A}_2,{\bf B}_1,{\bf B}_2\in \mathbb{F}_q^{m/2\times 1}$, the following sum rate is achievable by the separate encoding of the sources for the receiver to recover $\langle \,{\bf A},{\bf B} \,\rangle$ with a small probability of error: 
\begin{align}
\label{KW_sum_rate_for_inner_product}
\RKMsum = 2H({\bf U},\, {\bf V}, \, W) \ ,
\end{align}
where ${\bf U}, {\bf V}\in \mathbb{F}_q^{m/2\times 1}$ are vector variables, and $W\in \mathbb{F}_q$ is a random variable, and they satisfy the following relations:
\begin{align}
\label{UVW}
{\bf U}&={\bf A}_2\oplus_q{\bf B}_1 \ ,\nonumber\\
{\bf V}&={\bf A}_1\oplus_q{\bf B}_2 \ , \nonumber\\
W&={\bf A}_2^{\intercal} {\bf A}_1\oplus_q{\bf B}_1^{\intercal} {\bf B}_2 \ .
\end{align}
\end{prop}

\begin{proof}
See Appendix~\ref{App:prop:KW_sum_rate_for_inner_product}.
\end{proof}

\begin{figure}[t!]
\centering
\includegraphics[width=0.4\textwidth]{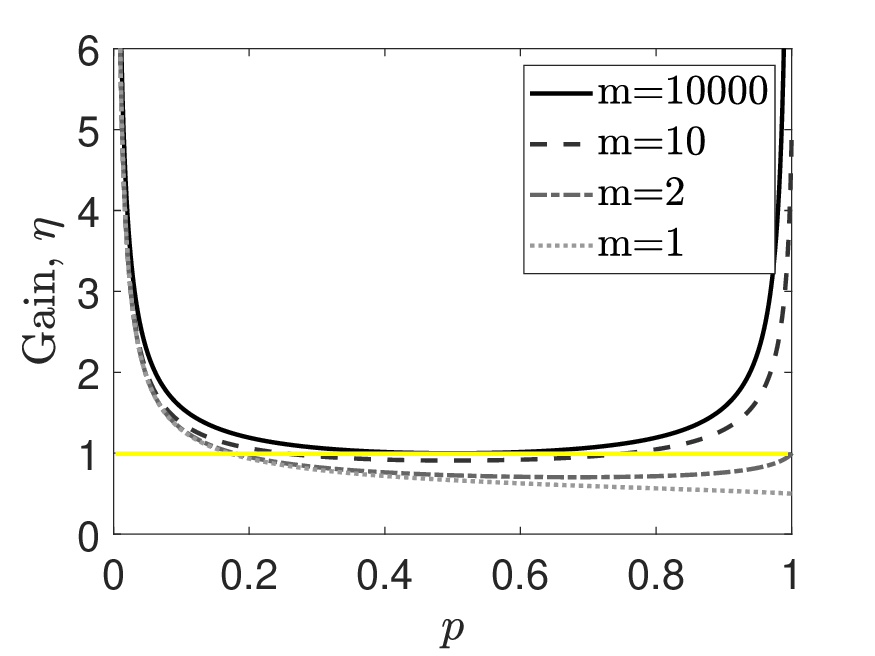}
\caption{Gain, $\eta$ (cf. (\ref{gain_DSBS_channel_KM}) in Corollary~\ref{cor:innerproduct_length_m_binary}. The flat (yellow) line marks $\eta=1$.}
\label{fig:distributedinnerproduct}
\end{figure}

In the scheme of Proposition~\ref{prop:KW_sum_rate_for_inner_product}, the receiver may not recover ${\bf A}$ and ${\bf B}$ in their entirety, from the modulo-$q$ additions ${\bf U}$, ${\bf V}$, and $W$, enabling secure distributed inner product computation. 

We next contrast the sum rate for computing $\langle \,{\bf A} ,{\bf B} \,\rangle$ given by Proposition~\ref{prop:KW_sum_rate_for_inner_product} with the sum rate of Slepian-Wolf in \cite{SlepWolf1973} to demonstrate the gap between $\RKMsum$ and $\RSWsum$. The following result indicates that the rate needed to compute $\langle \,{\bf A} ,{\bf B} \,\rangle$ may be substantially less than $H({\bf A}, {\bf B})$ for binary-valued sources.

\begin{cor}
\label{cor:innerproduct_length_m_binary}
Consider two sequences of correlated source vectors ${\bf A}\in\mathbb{F}_2^{m\times 1}$ and ${\bf B}\in\mathbb{F}_2^{m\times 1}$, with entries $a_{i}$ and $b_{i}$ that are i.i.d. across $i=1,\dots,m$ each, where $a_{i}\sim {\rm Bern}(1/2)$ and $b_{i}\sim {\rm Bern}(1/2)$ are correlated for a given $i=1,\dots,m$. Assume that ${\bf U}$ and ${\bf V}$, as defined in (\ref{UVW}), have entries $u_i, v_i \sim {\rm Bern}(p)$, that are i.i.d. across $i=1,\dots, m/2$. 

For this setting, the gain of the encoding technique in Proposition~\ref{prop:KW_sum_rate_for_inner_product} with a sum rate $\RKMsum$ given in (\ref{KW_sum_rate_for_inner_product}) over the sum rate $\RSWsum$ for lossless compression of the sources is
\begin{align}
\label{gain_DSBS_channel_KM}
\eta=\frac{\RSWsum}{\RKMsum}=\frac{m(1+h(p))}{2mh(p)+2(1-(1-p)^m)} \ .
\end{align}
\end{cor}

\begin{proof}
The proof follows from evaluating (\ref{KW_sum_rate_for_inner_product}) and contrasting it with $\RSWsum=H({\bf A},{\bf B})=m(1+h(p))$. For details, we refer the reader to Appendix~\ref{App:cor:innerproduct_length_m_binary}.
\end{proof}

Note from Corollary~\ref{cor:innerproduct_length_m_binary} that when $\eta>1$, the receiver can compute $\langle \,{\bf A} ,{\bf B} \,\rangle$ without recovering $({\bf A},{\bf B})$. It holds that
\begin{align}
\lim\limits_{p\to 0}\eta =\infty \ ,\quad \lim\limits_{p\to 1}\eta =\frac{m}{2} \ ,\quad \lim\limits_{m\to \infty}\eta= \frac{1+h(p)}{2h(p)} \ , 
\end{align}
where the limit $\lim\limits_{m\to \infty}\eta$ is the same as the gain for the DSBS model studied in \cite{korner1979encode}, which tends to infinity as $p\to\{0,1\}$.

\begin{figure*}[t!]
\centering
\includegraphics[width=0.325\textwidth]{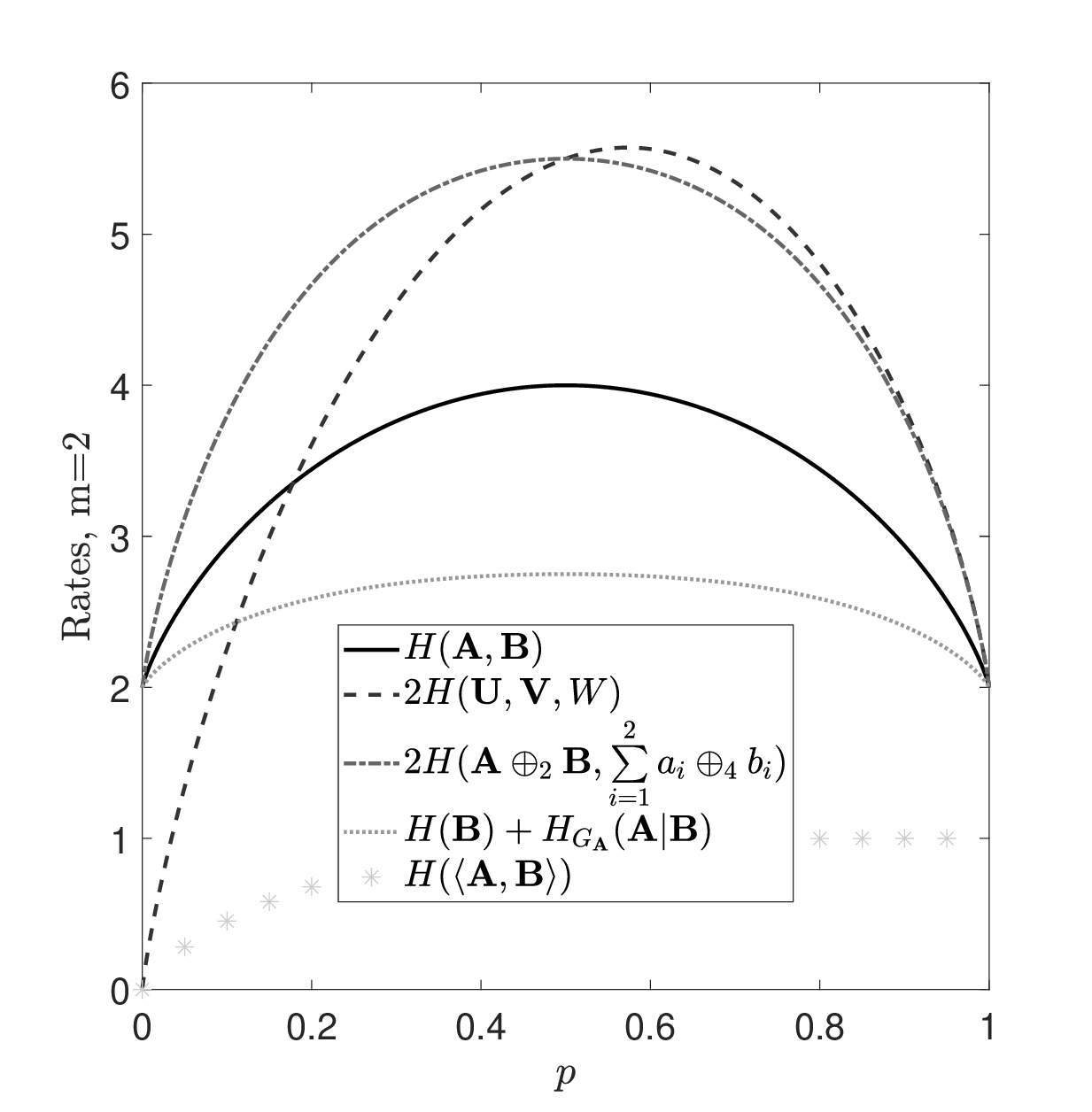}
\includegraphics[width=0.325\textwidth]{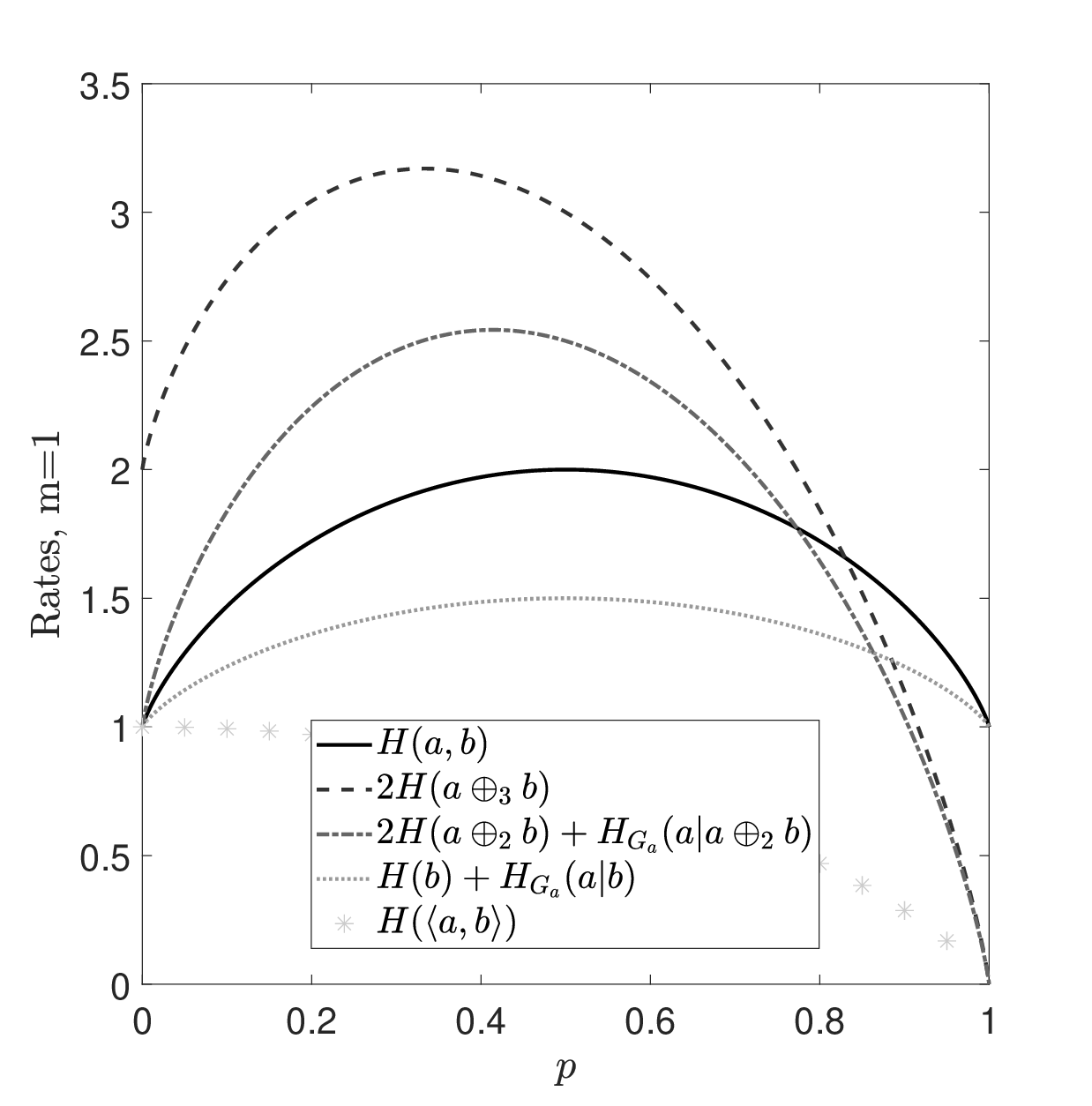}
\includegraphics[width=0.325\textwidth]{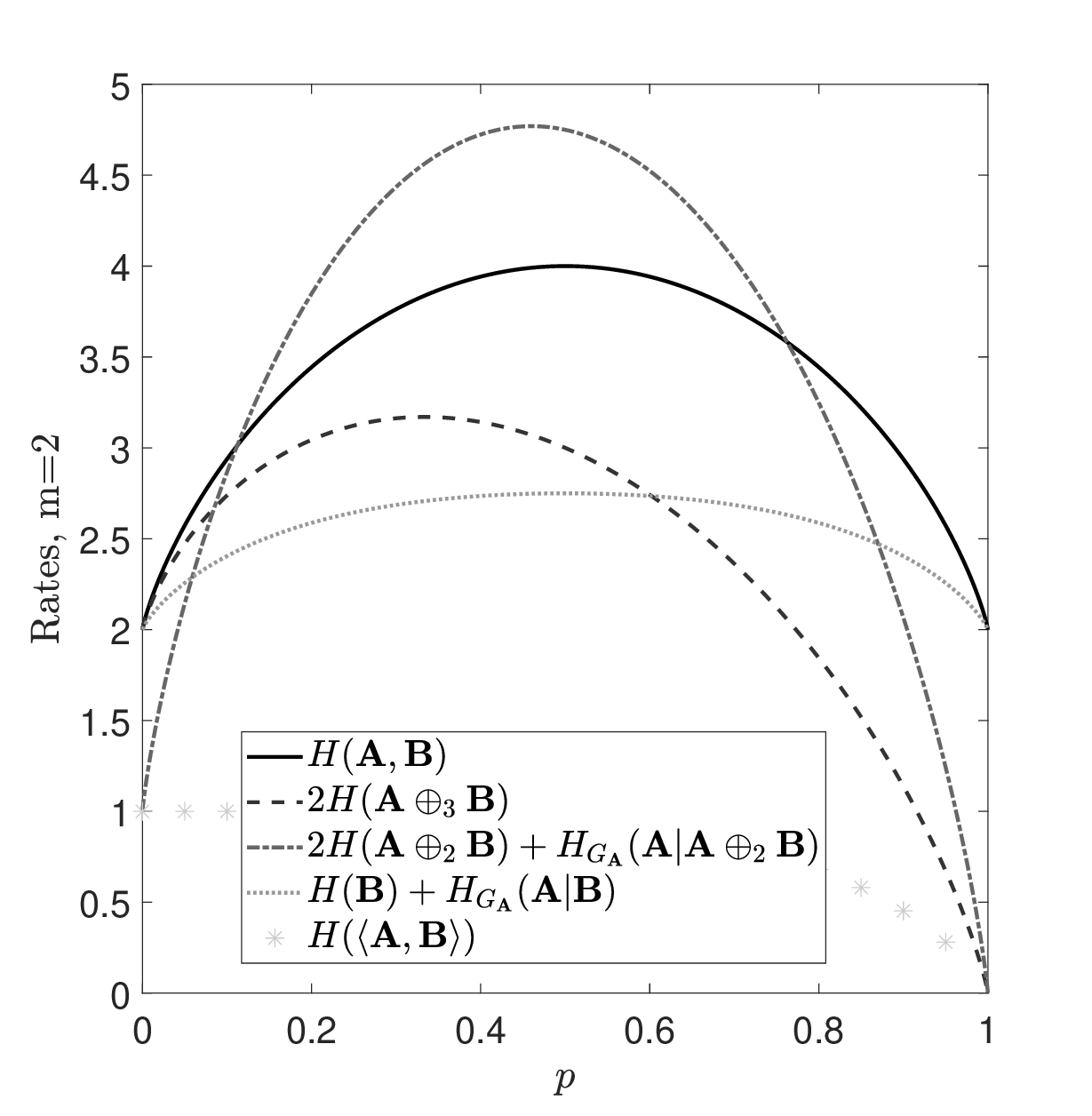}
\caption{Rate comparisons for various source PMFs. (Left) Corollary~\ref{cor:innerproduct_length_m_binary} for $m=2$. 
(Middle) $m=1$, and $a,b$ are DSBSs. (Right) $m=2$, and $\{a_i,b_i\}_{i=1}^2$ are DSBSs.}
\label{fig:innerproducts_general_comparison}
\end{figure*}

We illustrate the gain $\eta$ as a function of $(m,p)$ in Figure~\ref{fig:distributedinnerproduct}, indicating that $\RKMsum$ may be substantially less than the joint entropy of the sources for this special class of source PMFs.

It has been shown in \cite{krithivasan2011distributed}, \cite{pradhan2020algebraic} that via embedding the nonlinear function $D_k=a_kb_k$, where $a_k,b_k\in\mathbb{F}_q$, in a sufficiently large prime $\mathbb{F}_q$, the decoder can reconstruct $\tilde{D}^n=\{a_k\oplus_q b_k\}_{k=1}^n$, and hence, compute $D^n=\{a_kb_k\}_{k=1}^n$ with high probability. For instance, if $a_k,b_k\in\mathbb{F}_2$, we can reconstruct $D^n$ from $\tilde{D}^n=\{a_k\oplus_3 b_k\}_{k=1}^n$ using a sum rate of $\RSsum=2H({\bf A}\oplus_3{\bf B})$. 

Motivated by the notion of embedding in \cite{krithivasan2011distributed} and \cite{pradhan2020algebraic}, we next devise an achievability scheme for computing $\langle \,{\bf A},{\bf B} \,\rangle$, where the key idea is to compress the vector-wise embeddings of the sources vectors ${\bf A}$ and ${\bf B}$ via employing the linear structured encoding scheme of \cite{korner1979encode}, in contrast to entry-wise embeddings that require a sum rate of $\RSsum$ (cf.~\cite{krithivasan2011distributed} and \cite{pradhan2020algebraic}).

\begin{prop}
\label{prop:KW_sum_rate_for_inner_product_linear_source_mappings}
{\bf (Vector-wise embeddings followed by linear encoding  for distributed computation of $\langle \,{\bf A} ,{\bf B} \,\rangle$.)}
Given two sequences of vectors ${\bf A}\in\mathbb{F}_q^{m\times 1}$ and ${\bf B}\in\mathbb{F}_q^{m\times 1}$, generated by two correlated memoryless sources, with entries from a field $\mathbb{F}_q$ with $q>2$, restricting the source mappings to be linear, the following sum rate is achievable via the K\"orner-Marton's scheme to recover $\langle \,{\bf A} ,{\bf B} \,\rangle$ at the receiver with a small probability of error: 
\begin{align}
\label{linear_source_mapping_linear_encoding_qary}
\RSVsum=2H\Big(\Big\{ a_i\oplus_{r} b_i\Big\}_{i=1}^m,\, {\bigoplus_{i=1}^m} {_{_{{\qq}}} \, a_i^2\oplus_q b_i^2} \Big) \ ,
\end{align}
where $r=2(q-1)m$ for $m$ even, and $r=2(q-1)m+1$ for $m$ odd, respectively, and ${\bigoplus _{_{{\qq}}}}$ denotes a modulo-$q$ addition. 
\end{prop}

\begin{proof}
The proof follows from noting that the receiver, upon receiving $\big\{ a_i\oplus_{r} b_i\big\}_{i=1}^m$ and ${\bigoplus\limits_{i=1}^m} {_{_{{\qq}}} \, a_i^2\oplus_q b_i^2}$, can reconstruct
\begin{align}
2c=qk+\Big(\sum\limits_{i=1}^m (a_i\oplus_r b_i)^2 
-\, {\bigoplus\limits_{i=1}^m} {_{_{{\qq}}} \, (a_i^2\oplus_q b_i^2)\Big)}\hspace{-0.3cm}\mod q  \ ,\nonumber
\end{align}
where there is a unique $k\in \mathbb{F}_q$ for which $c=\langle \,{\bf A} ,{\bf B}\,\rangle \in \mathbb{F}_q$.

For details, we refer the reader to Appendix~\ref{App:prop:KW_sum_rate_for_inner_product_linear_source_mappings}.  
\end{proof}

We next describe a hybrid encoding scheme.  
Note that ${\bf A}^{\intercal}{\bf B}$ can be rewritten as $g({\bf A},{\bf Y})={\bf A}^{\intercal}({\bf Y}-{\bf A}) \, \mod q$ if ${\bf Y}={\bf A}\oplus_q {\bf B}$ is known. Exploiting  \emph{K\"orner's characteristic graphs} \cite{Kor73} to enable nonlinear encoding of ${\bf A}$\footnote{For a detailed description of characteristic graphs and their entropies, we refer the reader to \cite{feizi2014network,malak2022fractional,salehi2023achievable}.}, the minimum compression rate of ${\bf A}$ for computing $g({\bf A},{\bf Y})$ given side information ${\bf Y}$ equals the \emph{conditional graph entropy} of ${\bf A}$ given ${\bf Y}$, denoted by $H_{G_{\bf A}}({\bf A}\,\vert\, {\bf Y})$, as introduced by Orlitsky and Roche \cite{OR01}. Via concatenating the structured coding scheme of K\"orner and Marton \cite{korner1979encode} to first compute ${\bf Y}$ and then the unstructured coding model of Orlitsky and Roche \cite{OR01} to next determine $g({\bf A},{\bf Y})$, it is possible to achieve a sum rate $\RKMORsum=2H({\bf Y})+H_{G_{\bf A}}({\bf A}\,\vert\, {\bf Y})$. When ${\bf Y}={\bf B}$, the required rate is $H({\bf B})+H_{G_{\bf A}}({\bf A}\,\vert\, {\bf B})$, which is smaller than $\RSWsum$ because $H_{G_{\bf A}}({\bf A}\,\vert\, {\bf B})\leq H({\bf A}\,\vert\, {\bf B})$ \cite{OR01}. 

In Figure~\ref{fig:innerproducts_general_comparison}, we contrast the sum rate performance of Propositions~\ref{prop:KW_sum_rate_for_inner_product}-\ref{prop:KW_sum_rate_for_inner_product_linear_source_mappings} for distributed computing of $\langle \,{\bf A} ,{\bf B} \,\rangle$ of vectors ${\bf A},{\bf B}\in\mathbb{F}_2^{m\times 1}$ with a small probability of error, with existing unstructured and structured coding schemes.
In Figure~\ref{fig:innerproducts_general_comparison}-(Left), we use the PMF in Corollary~\ref{cor:innerproduct_length_m_binary} where $m=2$, i.e., the pairs $(a_1,b_2)$ and $(a_2,b_1)$ represent DSBSs, each with a crossover probability $p$. 
Note that the sum rates $\RSsum$ and $\RKMORsum$ perform poorly versus $\RSVsum$, and are not indicated. The sum rate $\RKMsum$ converges to $H(\langle \,{\bf A} ,{\bf B} \,\rangle)$ at low $p$.

In Figure~\ref{fig:innerproducts_general_comparison}-(Middle), we use $m=1$, and the pair $(a,b)$ is a DSBS with a crossover probability $p$. At low $p$, $\RKMORsum$ and $\RSWsum$ converge to $H(\langle \,a ,b \,\rangle)$  whereas $\RSsum$ performs poorly. For large $p$, structured coding yields low rates ($\RSsum$ and $\RKMORsum$).

In Figure~\ref{fig:innerproducts_general_comparison}-(Right), we use $m=2$, and the pairs $(a_1,b_1)$ and $(a_2,b_2)$ are DSBSs, each with a crossover probability $p$. The performance of $\RKMsum$ is worse than $\RSWsum$ and not shown. Similarly, $\RSVsum$ is higher than $\RKMORsum$, and is not indicated. For any given $p$ value, $\RSsum$ is always smaller than $\RSWsum$, and $\RKMORsum$ approaches $H(\langle \,{\bf A} ,{\bf B} \,\rangle)$ for small and large $p$.

\subsection{Distributed Computation of Symmetric Matrices}
\label{sec:achievability_results_symmetric_matrices}
We next consider a generalization of Proposition~\ref{prop:KW_sum_rate_for_inner_product} for distributed computing of inner products to distributed computing of a square symmetric matrix ${\bf D}=(D_{ij})\in\mathbb{F}_q^{l\times l}$, given by the product ${\bf D}={\bf A}^{\intercal}{\bf B}$, where ${\bf A}, {\bf B}\in\mathbb{F}_q^{m\times l}$, for $q>2$ and $l\geq 1$, and symmetry implies that $D_{ji}=D_{ij}$ for every $i,j$.

\begin{prop}
\label{prop:KW_sum_rate_for_symmetric_matrix_product}
{\bf (Computing symmetric matrices via distributed multiplication.)} Given two sequences of random matrices ${\bf A}=\begin{bmatrix}{\bf A}_1^{\intercal}&{\bf A}_2^{\intercal}\end{bmatrix}^{\intercal}\in\mathbb{F}_q^{m\times l}$ and ${\bf B}=\begin{bmatrix}{\bf B}_1^{\intercal} & {\bf B}_2^{\intercal}\end{bmatrix}^{\intercal}\in\mathbb{F}_q^{m\times l}$ generated by two correlated memoryless $q$-ary sources, respectively, where ${\bf A}_1,{\bf A}_2,{\bf B}_1,{\bf B}_2\in \mathbb{F}_q^{m/2\times l}$ and $q>2$, the achievable sum rate by the separate encoding of the sources for the receiver to recover the symmetric matrix ${\bf D}={\bf A}^{\intercal}{\bf B}\in\mathbb{F}_q^{l\times l}$ with vanishing error is given as 
\begin{align}
\label{KW_sum_rate_for_symmetric_matrix_product}
\RKMsum = 2H({\bf U},\, {\bf V}, \, {\bf W}) \ ,
\end{align}
where ${\bf U}={\bf A}_2\oplus_q{\bf B}_1\in \mathbb{F}_q^{m/2\times l}$, ${\bf V}={\bf A}_1\oplus_q{\bf B}_2\in \mathbb{F}_q^{m/2\times l}$, and ${\bf W}={\bf A}_2^{\intercal} {\bf A}_1\oplus_q{\bf B}_1^{\intercal} {\bf B}_2\in \mathbb{F}_q^{l\times l}$ are matrix variables.
\end{prop} 

\begin{proof}
The proof is on similar lines as Appendix~\ref{App:prop:KW_sum_rate_for_inner_product}. For details, we refer the reader to Appendix~\ref{App:prop:KW_sum_rate_for_symmetric_matrix_product}. 
\end{proof}

In the scheme of Proposition~\ref{prop:KW_sum_rate_for_symmetric_matrix_product}, similar to Proposition~\ref{prop:KW_sum_rate_for_inner_product}, the receiver may not recover ${\bf A}$ and ${\bf B}$ in their entirety, using ${\bf U}$, ${\bf V}$, and ${\bf W}$, hence rendering secure distributed compression of the sources feasible. To 
achieve information-theoretically secure distributed matrix multiplication, our future work includes exploiting the polynomial encoding scheme in \cite{chang2018capacity} that incorporates random matrices, 
and the techniques in \cite{jia2021capacity}, \cite{zhao2021expand}.

In Figure~\ref{fig:distributedgeneralmatrixproduct}-(Left), we showcase the sum rates $\RKMsum$ and $\RSWsum$ versus $p$ (in log scale) for distributed computing of symmetric matrices ${\bf D}={\bf A}^{\intercal}{\bf B}$ for $q=2$ under assumptions\footnote{For $q>2$, it is clear (from Appendix~\ref{App:prop:KW_sum_rate_for_symmetric_matrix_product} of Proposition~\ref{prop:KW_sum_rate_for_symmetric_matrix_product}) that the choices of ${\bf U}$, ${\bf V}$, and ${\bf W}$ guarantee the recovery of ${\bf A}^{\intercal}{\bf B}$ without assumptions i)-ii).}: i) ${\bf A}_1^{\intercal}{\bf B}_1={\bf B}_1^{\intercal}{\bf A}_1$, i.e., ${\bf W}={\bm 0}_{l\times l}$, and ii) ${\bf A}_2^{\intercal}{\bf A}_1={\bf B}_1^{\intercal}{\bf B}_2$. Because ${\bf D}$ is symmetric, i) and ii) lead to 
\begin{align}
{\bf U}^{\intercal}{\bf V}&=({\bf A}_2\oplus_2{\bf B}_1)^{\intercal}({\bf A}_1\oplus_2{\bf B}_2)\nonumber\\
&\overset{(a)}{=}{\bf A}_2^{\intercal}{\bf A}_1\oplus_2 {\bf A}_2^{\intercal}{\bf B}_2\oplus_2{\bf A}_1^{\intercal}{\bf B}_1\oplus_2{\bf B}_1^{\intercal}{\bf B}_2\nonumber\\
&\overset{(b)}{=}{\bf A}_1^{\intercal}{\bf B}_1\oplus_2 {\bf A}_2^{\intercal}{\bf B}_2={\bf D} \ ,\nonumber
\end{align}
where $(a)$ and $(b)$ follow from assumptions i) and ii), respectively. These assumptions ensure a rate gain of $\eta={\RSWsum}/{\RKMsum}$ that grows exponentially fast, as $p$ tends to $\{0,1\}$.

We next state a necessary condition for successful recovery of ${\bf D}={\bf A}^{\intercal}{\bf B}$, where ${\bf A}, {\bf B}\in\mathbb{F}_q^{m\times l}$, without recovering ${\bf A}$ and ${\bf B}$. This result holds true for any symmetric ${\bf D}\in\mathbb{F}_q^{l\times l}$.

\begin{prop}
\label{prop:KW_sum_rate_for_inner_product+vs_SW_sum_rate}
{\bf (A necessary condition for the nonrecovery of the sources in distributed computation of ${\bf A}^{\intercal}{\bf B}$.)} For distributed encoding of ${\bf A}, {\bf B}\in\mathbb{F}_q^{m\times l}$ to compute ${\bf D}={\bf A}^{\intercal}{\bf B}$, which is symmetric, the below condition ensures that the sum rate in (\ref{KW_sum_rate_for_inner_product}) is less than the achievable sum rate of \cite{SlepWolf1973}: 
\begin{align}
\label{lower_rate}
H({\bf A}^{\intercal}{\bf B})+H({\bf Q}\,\vert\,{\bf A}^{\intercal}{\bf B})< H({\bf A}\,\vert\,{\bf Q},{\bf A}^{\intercal}{\bf B}) \ ,
\end{align}
where ${\bf Q}=\begin{bmatrix}{\bf U} \\ {\bf V}\end{bmatrix}\in \mathbb{F}_q^{m\times 1}$, and $U$ and $V$ are defined in (\ref{UVW}).
\end{prop}

\begin{proof}
The inequality in (\ref{lower_rate}) can be obtained using the expansions of $\RKMsum$ and $\RSWsum$. 
For the technical steps, we refer the reader to Appendix~\ref{App:prop:KW_sum_rate_for_inner_product+vs_SW_sum_rate}.
\end{proof}

It is clear, using the expansion (\ref{KM_expansion}) of $\RKMsum$ in Appendix~\ref{App:prop:KW_sum_rate_for_inner_product+vs_SW_sum_rate},  that the sum rate $\RKMsum=2H({\bf A}^{\intercal}{\bf B})$ is achievable when $H({\bf Q} \,\vert\,{\bf A}^{\intercal}{\bf B})=0$ in (\ref{lower_rate}). In this case, (\ref{lower_rate}) implies that $H({\bf A}^{\intercal}{\bf B})< H({\bf A}\,\vert\,{\bf Q},{\bf A}^{\intercal}{\bf B})$, meaning that it is possible to recover the inner product ${\bf A}^{\intercal}{\bf B}$  while keeping ${\bf A}$ and ${\bf B}$ unknown to the receiver.

\subsection{Distributed Computation of Square Matrices}
\label{sec:achievability_results_general_matrices}

Recall that Proposition~\ref{prop:KW_sum_rate_for_symmetric_matrix_product} does not capture non-symmetric matrix products. We here consider distributed computing of a square matrix ${\bf D}=(D_{ij})\in\mathbb{F}_q^{l\times l}$, given by the product ${\bf D}={\bf A}^{\intercal}{\bf B}$, where ${\bf D}$ is not symmetric, and ${\bf A}, {\bf B}\in\mathbb{F}_q^{m\times l}$, for $l,\,m>1$, and $q>2$. The following proposition gives an achievable distributed encoding scheme of the sources ${\bf A}$ and ${\bf B}$ towards computing ${\bf D}$.

\begin{prop}
\label{prop:KW_sum_rate_for_general_matrix_product}
{\bf (Computing square matrices via distributed matrix multiplication.)} Given two sequences of random matrices ${\bf A}\in\mathbb{F}_q^{m\times l}$ and ${\bf B}\in\mathbb{F}_q^{m\times l}$ generated by two correlated memoryless $q$-ary sources, where $q>2$, the following sum rate is achievable by the separate encoding of the sources for the receiver to recover a general square matrix ${\bf D}={\bf A}^{\intercal}{\bf B}\in\mathbb{F}_q^{l\times l}$ with vanishing  error: 
\begin{align}
\label{KM_sum_rate_for_general_matrix_product}
\RKMsum = 2H(\{{\bf A}\oplus_q{\bf \tilde{B}}_j\}_{j=1}^l\, , \, \{{\bf \tilde{A}}_j^{\intercal}{\bf \tilde{A}}_j\}_{j=1}^l) \ ,
\end{align}
where we use the shorthand notation ${\bf \tilde{A}}_j^{\intercal}{\bf \tilde{A}}_j={\bf A}^{\intercal}{\bf A}\oplus_q{\bf \tilde{B}}_j^{\intercal}{\bf \tilde{B}}_j$ for $j\in \{1,2,\dots,l\}$, where ${\bf \tilde{B}}_j={\bf B}_j{\bm 1}_{1\times l}\in\mathbb{F}_q^{m\times l}$ are matrix variables,  ${\bm 1}_{1\times l}$ is a length $l$ row vector of all ones, where ${\bf B}=\begin{bmatrix}{\bf B}_1 & {\bf B}_2 & \hdots & {\bf B}_l\end{bmatrix}$ with ${\bf B}_j\in\mathbb{F}_q^{m\times 1}$ for $j\in \{1,2,\dots,l\}$.
\end{prop}

\begin{proof}
We refer the reader to Appendix~\ref{App:prop:KW_sum_rate_for_general_matrix_product}. 
\end{proof}

\begin{figure*}[t!]
\centering
\includegraphics[width=0.49\textwidth]{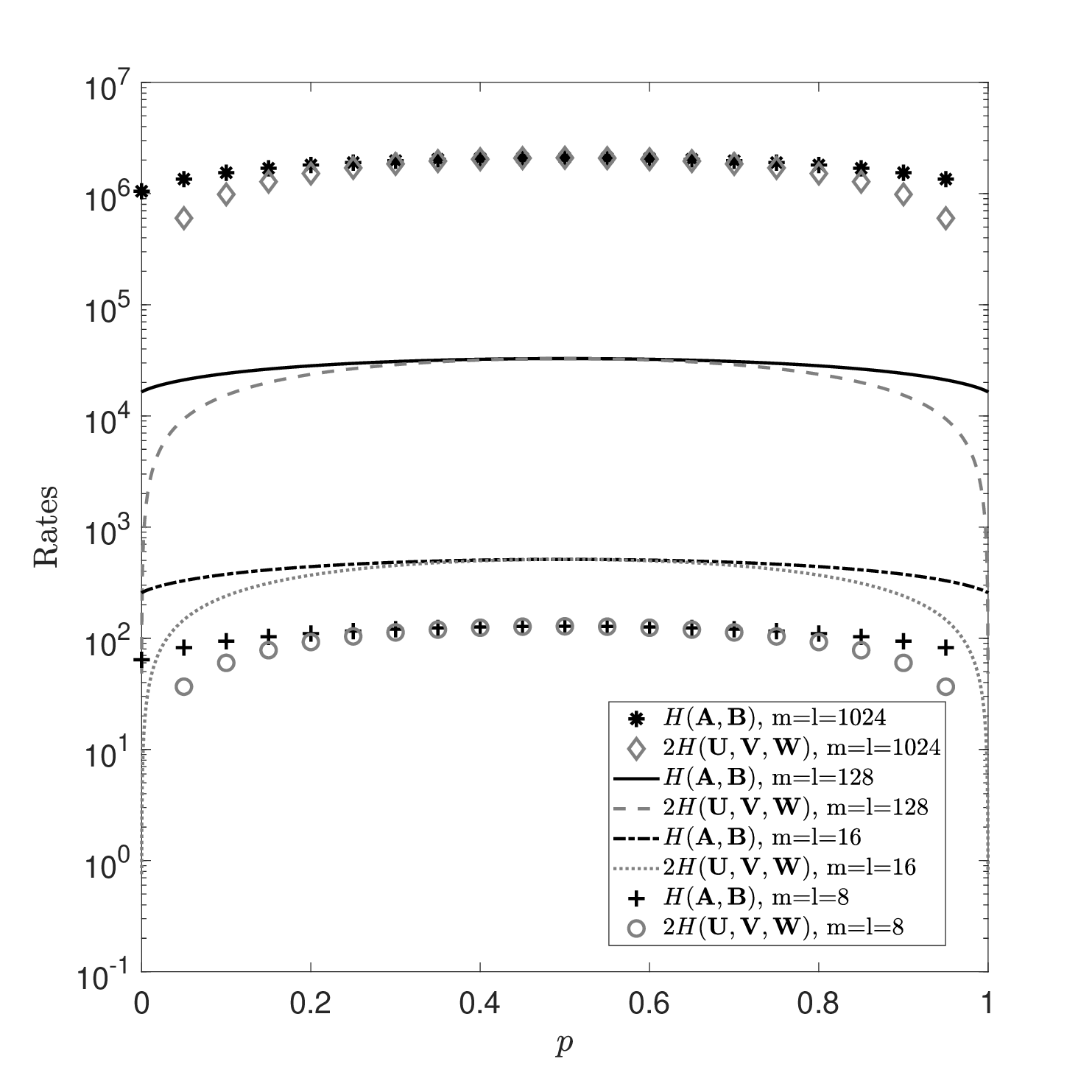}
\includegraphics[width=0.49\textwidth]{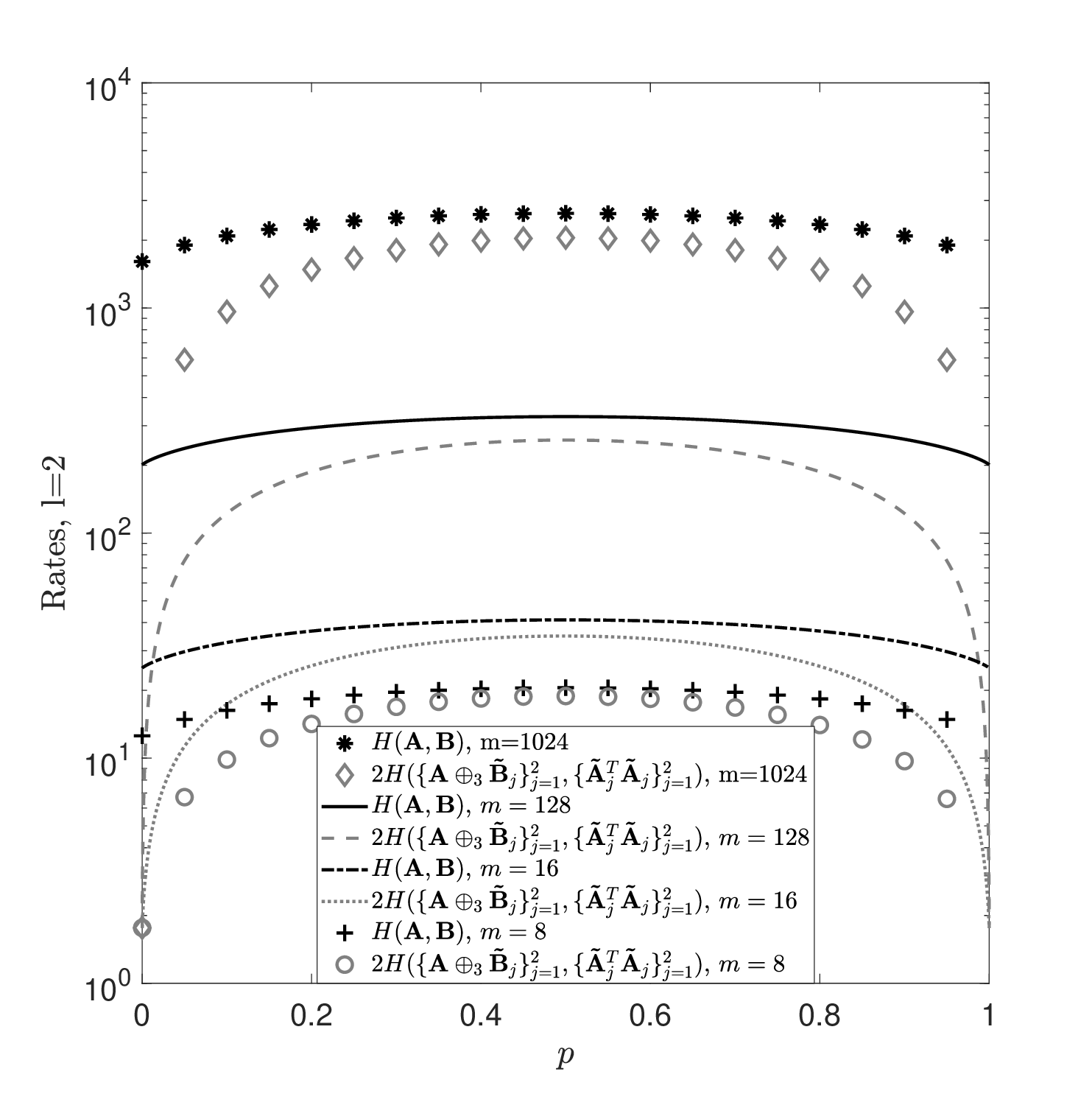}
\caption{Rate (in log scale) versus $p$ for distributed computing of (Left) symmetric matrices ${\bf A}^{\intercal}{\bf B}={\bf B}^{\intercal}{\bf A}$ via distributed multiplication of matrices ${\bf A}, {\bf B}\in \mathbb{F}_2^{m\times m}$ for different $m$, and (Right) square matrices via distributed matrix multiplication for different $m$ and $l=2$, where the joint source PMF is given in Corollary~\ref{cor:general_matrix_q3}.}
\label{fig:distributedgeneralmatrixproduct}
\end{figure*}

To demonstrate the performance of Proposition~\ref{prop:KW_sum_rate_for_general_matrix_product}, we next consider a corollary, where $l=2$, $m>1$, and with $q=3$.
\begin{cor} 
\label{cor:general_matrix_q3}
Given ${\bf A}=\begin{bmatrix} {\bf A}_1 &{\bf A}_2\end{bmatrix}\in\mathbb{F}_3^{m\times 2}$ and ${\bf B}=\begin{bmatrix} {\bf B}_1 &{\bf B}_2\end{bmatrix}\in\mathbb{F}_3^{m\times 2}$ with entries $a_{ij}\sim \big(\frac{1}{2}-\epsilon , 2\epsilon , \frac{1}{2}-\epsilon \big)$ for some $\epsilon\in\big[0,\frac{1}{2}\big]$, and $b_{i1}=b_{i2}=-a_{i2}$ that are i.i.d. across $i=1,\dots,m$ and the joint PMF of $a_{i1},b_{i1}$ satisfies
\begin{align}
P_{a_{i1},b_{i1}}=
\begin{bmatrix}
(\frac{1}{2}-\epsilon)(1-p)& (\frac{1}{2}-\epsilon)p & 0 \\
2\epsilon p & 0 & 2\epsilon(1-p)\\
0 & (\frac{1}{2}-\epsilon)(1-p)& (\frac{1}{2}-\epsilon)p
\end{bmatrix}\ .\nonumber
\end{align}
The sum rate for distributed encoding of $({\bf A},{\bf B})$ is given as
\begin{align}
\label{SW_sum_rate_for_general_matrix_product_q3}
\RSWsum=m(h(2\epsilon)+(1-2\epsilon)+h(p)) \ .
\end{align}

Exploiting Proposition~\ref{prop:KW_sum_rate_for_general_matrix_product} to compute ${\bf A}^{\intercal}{\bf B}$, we can achieve
\begin{align}
\label{KM_sum_rate_for_general_matrix_product_q3}
&\RKMsum\leq 2m h\Big(2\Big(\frac{1}{2}-\epsilon\Big)(1-p)+2\epsilon(1-p),\nonumber\\
&\hspace{2.1cm}2\Big(\frac{1}{2}-\epsilon\Big)p+2\epsilon p\Big)+2\log_2(3) \ .
\end{align}
\end{cor}

\begin{proof}
We refer the reader to Appendix~\ref{App:cor:general_matrix_q3}.
\end{proof}

In Figure~\ref{fig:distributedgeneralmatrixproduct}-(Right), we demonstrate the sum rate performance of Proposition~\ref{prop:KW_sum_rate_for_general_matrix_product} (in log scale) versus $p$ via contrasting the sum rates $H({\bf A},{\bf B})$ and (\ref{KM_sum_rate_for_general_matrix_product}) 
for the joint PMF model in Corollary~\ref{cor:general_matrix_q3}, where we assume that $\epsilon=0.2$. We observe that the rate gain  $\eta={\RSWsum}/{\RKMsum}$ grows exponentially fast, as $p$ tends to $\{0,1\}$.

{\bf Discussion.}
We proposed structured coding techniques for nonlinear mappings of distributed sources in a $q$-ary prime finite field to perform inner product-based matrix computation toward realizing distributed multiplication for special matrix classes, e.g., symmetric and square matrices, through imposing structural constraints on sources. 
Our future work includes the study of general matrix products, providing insights into the problems of distributed rank computation, trace computation, and low-rank matrix factorization, as well as the derivation of tighter achievability bounds for the distributed multiplication of general matrices and higher dimensional matrices or tensors.

\section*{Acknowledgment}
The author gratefully acknowledges the constructive discussions with Prof. Arun Padakandla at EURECOM.

\appendix

\subsection{Proof of Proposition \ref{prop:KW_sum_rate_for_inner_product}}
\label{App:prop:KW_sum_rate_for_inner_product}
The receiver aims to compute the inner product $\langle \,{\bf A} ,{\bf B} \,\rangle=\sum\nolimits_{i=1}^{m} a_i b_i$, with entries from a field of characteristic $q\geq 2$. Here, we focus on $q=2$, and the generalization to $q>2$ is straightforward \cite{han1987dichotomy}.

{\bf Encoding:}
Sources devise mappings $g_1: {\bf A}\to {\bf X}_1$ and $g_2: {\bf B}\to {\bf X}_2$, respectively, defined below, to determine the binary-valued column vectors
\begin{align}
\label{distributed_source_info}
{\bf X}_1&=g_1({\bf A})=\begin{bmatrix}{\bf A}_2 \\ {\bf A}_1 \\ {\bf A}_2^{\intercal} {\bf A}_1\end{bmatrix}\in \mathbb{F}_2^{(m+1) \times 1} \ ,\\
{\bf X}_2&=g_2({\bf B})=\begin{bmatrix}{\bf B}_1 \\ {\bf B}_2 \\ {\bf B}_1^{\intercal} {\bf B}_2\end{bmatrix}\in \mathbb{F}_2^{(m+1) \times 1} \ .
\end{align}
We denote by ${\bf X}_1^n, {\bf X}_2^n \in \mathbb{F}_2^{(m+1) \times n}$ the length $n$ source vectors. 
The encoders of the sources $\{{\bf X}_{1i}\}$ and $\{{\bf X}_{2i}\}$ are defined by functions $f_1: {\bf X}_1^n\to \mathcal{R}_{f_1}$ and $f_2: {\bf X}_2^n\to \mathcal{R}_{f_2}$, where $\mathcal{R}_{f_1}$ and $\mathcal{R}_{f_2}$ denote the ranges of $f_1$ and $f_2$, respectively. 
The pair of functions $(f_1,f_2)$ is called an $(n,\epsilon)$-coding scheme if there exists a function $\phi: \mathcal{R}_{f_1}\times \mathcal{R}_{f_2} \to \mathcal{Z}^n$ such that by letting 
\begin{align}
{\bf \hat{Z}}^n \triangleq \phi(f_1({\bf X}_1^n),f_2({\bf X}_2^n))\ , 
\end{align}
we have $\mathbb{P}({\bf \hat{Z}}^n \neq {\bf Z}^n)<\epsilon$. Here, ${\bf Z}$ is the modulo-two sum of ${\bf X}_1$ and ${\bf X}_2$, i.e., ${\bf Z}={\bf X}_1\oplus_2{\bf X}_2\in \mathbb{F}_2^{(m+1) \times 1}$. 

Our encoding scheme requires a well-known lemma of Elias \cite{gallager1968information}, which showed that linear codes achieve the capacity of binary symmetric channels, and its adaption to the problem of computing the modulo-two sum of DSBSs in \cite{korner1979encode}. Using a simple generalization of this result to vector variables, for fixed $\epsilon>0$ and for sufficiently large $n$, there exists a binary matrix ${\bf C}\in\mathbb{F}_2^{k\times n}$, where $f_1({\bf X}_1^n)\triangleq {\bf C}({\bf X}_1^n)={\bf C}\cdot ({\bf X}_1^n)^{\intercal}\in\mathbb{F}_2^{k\times (m+1)}$ and $f_2({\bf X}_2^n)\triangleq {\bf C}\cdot ({\bf X}_2^n)^{\intercal}\in\mathbb{F}_2^{k\times (m+1)}$ denote the modulo-two product of the matrix ${\bf C}$ with the transpose of the binary vector sequences ${\bf X}_1^n$ and ${\bf X}_2^n$, respectively, and a decoding function $\psi:\{0,1\}^{k\times (m+1)}\to \{0,1\}^{n\times (m+1)}$ that satisfy 
\begin{align}
\phi(f_1({\bf X}_1^n),f_2({\bf X}_2^n)) \triangleq \psi(f_1({\bf X}_1^n)+f_2({\bf X}_2^n))\nonumber
\end{align}
such that i) $k<n(H({\bf Z})+\epsilon)$, and ii) $\mathbb{P}(\psi({\bf C}({\bf Z}^n))\neq {\bf Z}^n)<\epsilon$. Hence, application of Elias's lemma  \cite{gallager1968information} and \cite{korner1979encode} yields that $({\bf C},{\bf C})$ is an $(n,\epsilon)$-coding scheme.

{\bf Decoding:} 
Exploiting the achievability result of K\"orner-Marton \cite{korner1979encode}, the sum rate needed for the receiver to recover the vector sequence ${\bf Z}^n={\bf X}_1^n\oplus_2{\bf X}_2^n\in \mathbb{F}_2^{(m+1) \times n}$ with a vanishing error probability can be determined as \cite{korner1979encode}:
\begin{align}
\label{KW_sum_rate_for_inner_product_achievable}
\RKMsum = 2H({\bf U},\, {\bf V}, \, W) \ .
\end{align}
Using $\psi$ given in ii), the receiver can recover ${\bf Z}^n$. However, lossless decoding of ${\bf X}_1^n$ and ${\bf X}_2^n$ is not guaranteed.

We next show that the sum rate in (\ref{KW_sum_rate_for_inner_product_achievable}) is sufficient to recover $\langle \,{\bf A} ,{\bf B} \,\rangle$.  
Using ${\bf \hat{Z}}^n$, the receiver computes
\begin{align}
\label{sufficiency_inner_product}
{\bf U}^{\intercal} {\bf V}-W&=({\bf A}_2^{\intercal}\oplus_q{\bf B}_1^{\intercal})({\bf A}_1\oplus_q{\bf B}_2)\nonumber\\
&-({\bf A}_2^{\intercal} {\bf A}_1\oplus_q{\bf B}_1^{\intercal} {\bf B}_2)&\nonumber\\
&={\bf B}_1^{\intercal}{\bf A}_1\oplus_q{\bf A}_2^{\intercal}{\bf B}_2\nonumber\\
&\overset{(a)}{=}{\bf A}_1^{\intercal}{\bf B}_1\oplus_q{\bf A}_2^{\intercal}{\bf B}_2=\langle \,{\bf A} ,{\bf B} \,\rangle & \  ,
\end{align}
where $(a)$ follows from ${\bf B}_1^{\intercal}{\bf A}_1={\bf A}_1^{\intercal}{\bf B}_1\in\mathbb{F}_q$.

Hence, (\ref{KW_sum_rate_for_inner_product_achievable}) with (\ref{sufficiency_inner_product}) gives the achievability result we seek.

\subsection{Proof of Corollary~\ref{cor:innerproduct_length_m_binary}}
\label{App:cor:innerproduct_length_m_binary}
Employing the definitions of ${\bf X}_1$ and ${\bf X}_2$ in (\ref{distributed_source_info}), (\ref{KW_sum_rate_for_inner_product}) and \cite{korner1979encode}, we can determine the sum rate needed for the receiver to recover $({\bf U},\, {\bf V}, \, W)$ in an asymptotic manner:
\begin{align}
\label{KW_sum_rate_inner_product}
\RKMsum &= 2H({\bf U},\, {\bf V}, \, W) \nonumber\\
&=2H({\bf U})+2H({\bf V})+2H({\bf A}_2^{\intercal} {\bf A}_1\oplus_2{\bf B}_1^{\intercal} {\bf B}_2 \, \vert \, {\bf U}, {\bf V})\nonumber\\
&=2\frac{m}{2}h(p)+2\frac{m}{2}h(p)+2H({\bf A}_2^{\intercal} {\bf A}_1\oplus_2{\bf B}_1^{\intercal} {\bf B}_2 \, \vert \, {\bf U}, {\bf V})\nonumber\\
&\overset{(a)}{=}2mh(p)+2H({\bf U}^{\intercal}{\bf A}_1\oplus_2 {\bf A}_2^{\intercal}{\bf V}\, \vert \, {\bf U}, {\bf V})\nonumber\\
&\overset{(b)}{=}2mh(p)+2H({\bf Q}^{\intercal}{\bf A} \, \vert  \, {\bf U}, {\bf V}) \nonumber\\
&\overset{(c)}{=}2mh(p)+2H\Big(\sum\limits_{i=1}^{m} a_i z_i \,\vert\, {\bf Q}\Big)\nonumber\\
&\overset{(d)}{=}2mh(p)+2\sum\limits_{j=1}^m \binom{m}{j} p^j(1-p)^{m-j} H\Big(\sum\limits_{i=1}^j a_i\Big)\nonumber\\
&\overset{(e)}{=}2mh(p)+2(1-(1-p)^m) \ ,
\end{align}
where $(a)$ follows from employing the relations ${\bf U}={\bf A}_2\oplus_2{\bf B}_1=\begin{bmatrix}u_1 & u_2 & \hdots & u_{m/2}\end{bmatrix}^{\intercal}$ and ${\bf V}={\bf A}_1\oplus_2{\bf B}_2=\begin{bmatrix}v_1 & v_2 & \hdots & v_{m/2}\end{bmatrix}^{\intercal}$, and simplification using conditioning, $(b)$ follows from employing ${\bf Q}=\begin{bmatrix}{\bf U} \\ {\bf V}\end{bmatrix}\in\mathbb{F}_2^{m\times 1}$ and ${\bf A}_2^{\intercal}{\bf V}={\bf V}^{\intercal}{\bf A}_2$, 
$(c)$ from the definition of ${\bf Q}$, and $(d)$ from $H({\bf Q}^{\intercal}{\bf A} \,\vert\, {\bf Q})\leq H({\bf Q}^{\intercal}{\bf A} \,\vert\, {\bf Q}^{\intercal}{\bm 1}_m)$, where ${\bf Q}^{\intercal}{\bm 1}_m=\sum\limits_{i=1}^m y_i \sim B(m,p)$, such that $H({\bf Q}^{\intercal}{\bf A} \,\vert\, {\bf Q}^{\intercal}{\bm 1}_m=j)=H\Big(\sum\limits_{i=1}^j a_i\Big)$, $j\geq 1$ exploiting that $a_i\overset{i.i.d.}{\sim} {\rm Bern}(1/2)$, and $H({\bf Q}^{\intercal}{\bf A} \,\vert\, {\bf Q}^{\intercal}{\bm 1}_m=0)=0$. Finally, incorporating $H\Big(\sum\limits_{i=1}^j a_i\Big)=1$, $j\geq 1$ we obtain $(e)$.

The encoding rate for asymptotic lossless compression of ${\bf A}$ and ${\bf B}$ is given by the Slepian-Wolf theorem \cite{SlepWolf1973}:
\begin{align}
\label{SW_sum_rate_inner_product}
\RSWsum&=H({\bf A},{\bf B})\nonumber\\
&\overset{(a)}{=}2H({\bf A}_1,{\bf B}_2)\nonumber\\
&=2\frac{m}{2}(1+h(p))\nonumber\\
&=m(1+h(p)) \ ,
\end{align}
where $(a)$ follows from using ${\bf A}=\begin{bmatrix}{\bf A}_1\\{\bf A}_2\end{bmatrix}$ and ${\bf B}=\begin{bmatrix}{\bf B}_1\\{\bf B}_2\end{bmatrix}$, with ${\bf A}_1 \independent {\bf A}_2$,  ${\bf B}_1 \independent {\bf B}_2$, and ${\bf A}$ and ${\bf B}$ having i.i.d. entries.

From (\ref{KW_sum_rate_inner_product}) and (\ref{SW_sum_rate_inner_product}), $\eta={\RSWsum}/{\RKMsum}$ is given by (\ref{gain_DSBS_channel_KM}).

\subsection{Proof of Proposition~\ref{prop:KW_sum_rate_for_inner_product_linear_source_mappings}}
\label{App:prop:KW_sum_rate_for_inner_product_linear_source_mappings}
Here, we provide a sketch of the proof. For binary sources, i.e., $q=2$, it is easy to verify that
\begin{align}
\langle \,{\bf A} ,{\bf B} \,\rangle=\Big\lfloor \frac{1}{2}\Big(\sum\limits_{i=1}^m a_i\oplus_{r} b_i-\sum\limits_{i=1}^m (a_i\oplus_2 b_i) \Big)\Big\rfloor \mod 2 \ ,\nonumber
\end{align}
where $r=2m$ for $m$ even, and $r=2m+1$ for $m$ odd, respectively. Hence, the following sum rate is achievable:
\begin{align}
\label{linear_source_mapping_linear_encoding_binary}
\RSsum=2H\Big(\sum\limits_{i=1}^m a_i\oplus_{r} b_i,\, \{a_i\oplus_2 b_i\}_{i=1}^m \Big) \ .
\end{align}
When the data is generated by two correlated memoryless $q$-ary sources for $q\geq 2$, it is possible to achieve a sum rate
\begin{align}
\RSsum=2H\Big(\Big\{ a_i\oplus_{r} b_i\Big\}_{i=1}^m,\, {\bigoplus_{i=1}^m} {_{_{{\qq}}} \, a_i^2\oplus_q b_i^2} \Big) \ ,\nonumber
\end{align}
where $r=2(q-1)m$ and $r=2(q-1)m+1$ for even and odd $m$, respectively. Upon receiving $\big\{ a_i\oplus_{r} b_i\big\}_{i=1}^m$ and ${\bigoplus\limits_{i=1}^m} {_{_{{\qq}}} \, a_i^2\oplus_q b_i^2}$, the receiver can reconstruct
\begin{align}
2c=qk+\Big(\sum\limits_{i=1}^m (a_i\oplus_r b_i)^2 
-\, {\bigoplus\limits_{i=1}^m} {_{_{{\qq}}} \, (a_i^2\oplus_q b_i^2)\Big)}\hspace{-0.3cm}\mod q  \ ,\nonumber
\end{align}
where there is a unique $k\in \mathbb{F}_q$ for which $c=\langle \,{\bf A} ,{\bf B}\,\rangle \in \mathbb{F}_q$.

\subsection{Proof of Proposition~\ref{prop:KW_sum_rate_for_symmetric_matrix_product}}
\label{App:prop:KW_sum_rate_for_symmetric_matrix_product}

Given two sequences of random matrices ${\bf A},\, {\bf B}\in\mathbb{F}_q^{m\times l}$, the receiver aims to compute ${\bf A}^{\intercal}{\bf B}={\bf B}^{\intercal}{\bf A}\in \mathbb{F}_q^{l\times l}$.

{\bf Encoding:}
Each source uses mappings $g_1: {\bf A}\to {\bf X}_1$ and $g_2: {\bf B}\to {\bf X}_2$, to determine the respective matrices:
\begin{align}
\label{distributed_source_info_matrices}
{\bf X}_1&=g_1({\bf A})=\begin{bmatrix}{\bf A}_2 \\ {\bf A}_1 \\ {\bf A}_2^{\intercal} {\bf A}_1\end{bmatrix}\in \mathbb{F}_q^{(m+l) \times l} \ ,\\
{\bf X}_2&=g_2({\bf B})=\begin{bmatrix}{\bf B}_1 \\ {\bf B}_2 \\ {\bf B}_1^{\intercal} {\bf B}_2\end{bmatrix}\in \mathbb{F}_q^{(m+l) \times l} \ .
\end{align} 

Following the steps of Appendix~\ref{App:prop:KW_sum_rate_for_inner_product}, there exists an $(n,\epsilon)$-coding scheme $({\bf C},{\bf C})$ for a matrix ${\bf C}\in\mathbb{F}_q^{k\times n}$, where $k=(m+l)\times l$, for decoding ${\bf Z}={\bf X}_1\oplus_q{\bf X}_2\in \mathbb{F}_q^{(m+l) \times l}$ with a small probability of error.

{\bf Decoding:} 
Exploiting the achievability result of K\"orner-Marton \cite{korner1979encode}, the sum rate needed for the receiver to recover the matrix sequence ${\bf Z}^n={\bf X}_1^n\oplus_q{\bf X}_2^n$ with a vanishing error probability can be determined as \cite{korner1979encode}:
\begin{align}
\label{KW_sum_rate_for_symmetric_matrix_product_proof}
\RKMsum = 2H({\bf U},\, {\bf V}, \, {\bf W}) \ .
\end{align}

For $q>2$, given a symmetric matrix ${\bf D}={\bf A}^{\intercal}{\bf B}\in\mathbb{F}_q^{l\times l}$, the following relation holds:
\begin{align}
\label{symmetric_matrix}
{\bf D}=\frac{1}{2}({\bf D}\oplus_q{\bf D}^{\intercal}) = {\bf A}^{\intercal}{\bf B}={\bf B}^{\intercal}{\bf A} \ .
\end{align} 

Using ${\bf \hat{Z}}^n$, the receiver computes
\begin{align}
\label{sufficiency_symmetric_matrix_product}
&\frac{1}{2}(({\bf U}^{\intercal} {\bf V}-{\bf W})\oplus_q({\bf U}^{\intercal} {\bf V}-{\bf W})^{\intercal})\nonumber\\
&=\frac{1}{2}(({\bf A}_2^{\intercal}\oplus_q{\bf B}_1^{\intercal})({\bf A}_1\oplus_q{\bf B}_2)-({\bf A}_2^{\intercal} {\bf A}_1\oplus_q{\bf B}_1^{\intercal} {\bf B}_2)\nonumber\\
&\hspace{0.45cm}\oplus_q({\bf A}_1^{\intercal}\oplus_q{\bf B}_2^{\intercal})({\bf A}_2\oplus_q{\bf B}_1)-({\bf A}_1^{\intercal} {\bf A}_2\oplus_q{\bf B}_2^{\intercal} {\bf B}_1))\nonumber\\
&=\frac{1}{2}(({\bf B}_1^{\intercal}{\bf A}_1\oplus_q{\bf A}_2^{\intercal}{\bf B}_2)\oplus_q({\bf A}_1^{\intercal}{\bf B}_1\oplus_q {\bf B}_2^{\intercal}{\bf A}_2))\nonumber\\
&=\frac{1}{2}(({\bf A}_1^{\intercal}{\bf B}_1\oplus_q{\bf A}_2^{\intercal}{\bf B}_2)\oplus_q({\bf A}_1^{\intercal}{\bf B}_1\oplus_q{\bf A}_2^{\intercal}{\bf B}_2)^{\intercal})\nonumber\\
&=\frac{1}{2}({\bf A}^{\intercal}{\bf B}\oplus_q ({\bf A}^{\intercal}{\bf B})^{\intercal})={\bf A}^{\intercal}{\bf B}\ ,
\end{align}
where the last equality follows from that ${\bf D}={\bf A}^{\intercal}{\bf B}$ is a symmetric matrix and it satisfies (\ref{symmetric_matrix}).

Combining (\ref{KW_sum_rate_for_symmetric_matrix_product_proof}) with (\ref{sufficiency_symmetric_matrix_product}) gives the achievability result.

\subsection{Proof of Proposition~\ref{prop:KW_sum_rate_for_inner_product+vs_SW_sum_rate}}
\label{App:prop:KW_sum_rate_for_inner_product+vs_SW_sum_rate}

We first show that the encoding scheme of Proposition \ref{prop:KW_sum_rate_for_inner_product} does not allow the recovery of ${\bf A},{\bf B}\in\mathbb{F}_q^{m\times 1}$ by the receiver, i.e., $H({\bf A},{\bf B}\,\vert\,{\bf A}^{\intercal}{\bf B},{\bf Q})>0$, for $m>1$. 

The receiver can recover ${\bf Q}=\begin{bmatrix}{\bf U} \\ {\bf V}\end{bmatrix}\in \mathbb{F}_q^{m\times 1}$ and $W\in\mathbb{F}_q$ with a small probability of error. The extra rate needed from the encoders for the receiver to determine ${\bf A},\,{\bf B}\in \mathbb{F}_q^{m\times 1}$ is 
\begin{align}
H({\bf A},{\bf B}\,\vert\,{\bf A}^{\intercal}{\bf B},{\bf Q})&\overset{(a)}{=}H({\bf A},{\bf B},{\bf A}^{\intercal}{\bf B},{\bf Q})-H({\bf A}^{\intercal}{\bf B},{\bf Q})\nonumber\\
&\overset{(b)}{=}H({\bf A},{\bf B},{\bf Q})-H({\bf A}^{\intercal}{\bf B},{\bf Q})\nonumber\\
&\overset{(c)}{=}H({\bf A},{\bf Q})-H({\bf A}^{\intercal}{\bf Q},{\bf Q})\nonumber\\
&\overset{(d)}{=}H({\bf A},{\bf A}^{\intercal}{\bf Q},{\bf Q})-H({\bf A}^{\intercal}{\bf Q},{\bf Q})\nonumber\\
&\overset{(e)}{=}H({\bf A}\,\vert\,{\bf A}^{\intercal}{\bf Q},{\bf Q})\overset{(f)}{\geq} 0 \ ,\nonumber
\end{align}
where $(a)$ follows from using the definition of conditional entropy and that the receiver can compute ${\bf A}^{\intercal}{\bf B}={\bf U}^{\intercal} {\bf V}-W$ from ${\bf Q}$ and $W$, $(b)$ from $H({\bf A}^{\intercal}{\bf B}\,\vert\,{\bf Q},{\bf A},{\bf B})=0$, $(c)$ from $H({\bf B}\,\vert\,{\bf A},{\bf Q})=0$ and ${\bf A}^{\intercal}{\bf B}={\bf A}^{\intercal}{\bf Q}$ given ${\bf Q}$, $(d)$ from $H({\bf A}^{\intercal}{\bf Q}\,\vert\,{\bf A},{\bf Q})=0$, $(e)$ from employing the definition of conditional entropy, and $(f)$ holds with equality if the function $g({\bf A},{\bf Q})={\bf A}^{\intercal}{\bf Q}$ is partially invertible, i.e., $H({\bf A}\,\vert\,g({\bf A},{\bf Q}), {\bf Q} ) = 0$, e.g., when $g$ is the arithmetic sum or the modulo sum of the two vectors. Hence, the inequality in $(f)$ is strict for inferring ${\bf A},\,{\bf B}\in \mathbb{F}_q^{m\times 1}$ from ${\bf Q}$ and $W$.

We next prove the main result of the proposition. 
Given matrix variables ${\bf A}, {\bf B}\in\mathbb{F}_q^{m\times l}$ such that ${\bf D}={\bf A}^{\intercal}{\bf B}$ is symmetric, letting ${\bf Q}=\begin{bmatrix}{\bf U} \\ {\bf V}\end{bmatrix}$, we first expand $\RKMsum$ as 
\begin{align}
\label{KM_expansion}
\RKMsum&=2H({\bf U},\, {\bf V}, \, W)\nonumber\\
&=2H({\bf U},\, {\bf V}, \, {\bf A}^{\intercal}{\bf B})\nonumber\\
&=2H({\bf A}^{\intercal}{\bf B})+2H({\bf Q}\,\vert\,{\bf A}^{\intercal}{\bf B})\ .
\end{align}

We next expand $\RSWsum$ as 
\begin{align}
\label{SW_expansion}
&\RSWsum=H({\bf A},{\bf B})=H({\bf A},{\bf B},{\bf U},\, {\bf V}, \, W,{\bf A}^{\intercal}{\bf B})\nonumber\\
&=H({\bf U},\, {\bf V}, \, {\bf A}^{\intercal}{\bf B})+H({\bf A}_1,{\bf B}_1,{\bf A}_2,{\bf B}_2\,\vert\,{\bf A}_2\oplus_q{\bf B}_1,\, \nonumber\\
&\hspace{3.8cm}{\bf A}_1\oplus_q{\bf B}_2,{\bf A}_1^{\intercal}{\bf B}_1\oplus_q{\bf A}_2^{\intercal}{\bf B}_2)\nonumber\\
&=H({\bf U},\, {\bf V}, \, {\bf A}^{\intercal}{\bf B})+H({\bf A}_1,{\bf A}_2\,\vert\,{\bf U},\, {\bf V},{\bf A}_1^{\intercal}({\bf U}\oplus_q{\bf A}_2)\nonumber\\
&\hspace{5cm}\oplus_q{\bf A}_2^{\intercal}({\bf V}\oplus_q{\bf A}_1))\nonumber\\
&=H({\bf Q}, \, {\bf A}^{\intercal}{\bf B})+H({\bf A}\,\vert\,{\bf Q},{\bf A}^{\intercal}{\bf Q}\oplus_q{\bf A}_1^{\intercal}{\bf A}_2\oplus_q{\bf A}_2^{\intercal}{\bf A}_1)\nonumber\\
&=H({\bf A}^{\intercal}{\bf B})+H({\bf Q} \,\vert\, {\bf A}^{\intercal}{\bf B})\nonumber\\
&+H({\bf A}\,\vert\,{\bf Q},{\bf A}^{\intercal}{\bf Q}\oplus_q{\bf A}_1^{\intercal}{\bf A}_2\oplus_q{\bf A}_2^{\intercal}{\bf A}_1)\ ,
\end{align}
where it is easy to observe that ${\bf A}^{\intercal}{\bf B}={\bf A}_1^{\intercal}({\bf U}+{\bf A}_2)+{\bf A}_2^{\intercal}({\bf V}+{\bf A}_1)={\bf A}^{\intercal}{\bf Q}\oplus_q{\bf A}_1^{\intercal}{\bf A}_2\oplus_q{\bf A}_2^{\intercal}{\bf A}_1$.

Similarly, via exploiting $\tilde{\bf Q}=\begin{bmatrix}{\bf V} \\ {\bf U}\end{bmatrix}$, we can show that
\begin{align}
\label{SW_expansion2}
\RSWsum&
=H({\bf A}^{\intercal}{\bf B})+H({\bf Q}\,\vert\, {\bf A}^{\intercal}{\bf B})\nonumber\\
&+H({\bf B}\,\vert\,\tilde{\bf Q},\tilde{\bf Q}^{\intercal}{\bf B}\oplus_q{\bf B}_1^{\intercal}{\bf B}_2+{\bf B}_2^{\intercal}{\bf B}_1)\ ,
\end{align} 
where ${\bf A}^{\intercal}{\bf B}
=\tilde{\bf Q}^{\intercal}{\bf B}\oplus_q{\bf B}_1^{\intercal}{\bf B}_2\oplus_q{\bf B}_2^{\intercal}{\bf B}_1$.

From (\ref{SW_expansion}) and (\ref{SW_expansion2}), we note that $H({\bf A}\,\vert\,{\bf A}^{\intercal}{\bf Q})=H({\bf B}\,\vert\,\tilde{\bf Q}^{\intercal}{\bf B})$. 
Contrasting (\ref{KM_expansion}) with (\ref{SW_expansion}), the condition
\begin{align}
H({\bf A}^{\intercal}{\bf B})+H({\bf Q}\,\vert\,{\bf A}^{\intercal}{\bf B})&< H({\bf A}\,\vert\,{\bf Q},{\bf A}^{\intercal}{\bf B})\nonumber\\
&=H({\bf B}\,\vert\,\tilde{\bf Q},{\bf A}^{\intercal}{\bf B})\nonumber
\end{align}
ensures that $\RKMsum<\RSWsum$. 

When $l=1$, we have ${\bf A}_1^{\intercal}{\bf A}_2={\bf A}_2^{\intercal}{\bf A}_1$ and ${\bf B}_1^{\intercal}{\bf B}_2={\bf B}_2^{\intercal}{\bf B}_1$, hence, the above condition is equivalent to
\begin{align}
H({\bf A}^{\intercal}{\bf B})+H({\bf Q}\,\vert\,{\bf A}^{\intercal}{\bf B})&<H({\bf A}\,\vert\,{\bf Q},{\bf A}^{\intercal}{\bf Q})\nonumber\\
&=H({\bf B}\,\vert\,\tilde{\bf Q},\tilde{\bf Q}^{\intercal}{\bf B})\ .\nonumber
\end{align}

\subsection{Proof of Proposition~\ref{prop:KW_sum_rate_for_general_matrix_product}}
\label{App:prop:KW_sum_rate_for_general_matrix_product}

Note that ${\bf D}_j={\bf A}^{\intercal}{\bf B}_j$ for $j\in\{1,\dots, l\}$, where ${\bf D}_j=\begin{bmatrix} d_{1j} & d_{2j} & \hdots & d_{lj}
\end{bmatrix}^{\intercal}\in\mathbb{F}_q^{l\times 1}$. 
Following the steps in Appendix~\ref{App:prop:KW_sum_rate_for_inner_product}, the receiver can recover $\{{\bf A}\oplus_q{\bf \tilde{B}}_j\}_{j=1}^l\, , \, \{{\bf A}^{\intercal}{\bf A}\oplus_q{\bf \tilde{B}}_j^{\intercal}{\bf \tilde{B}}_j\}_{j=1}^l$, and then compute the following $l\times l$ matrix:
\begin{align}
&({\bf A}\oplus_q{\bf \tilde{B}}_j)^{\intercal} ({\bf A}\oplus_q{\bf \tilde{B}}_j)-({\bf A}^{\intercal}{\bf A}\oplus_q{\bf \tilde{B}}_j^{\intercal}{\bf \tilde{B}}_j)\nonumber\\
&={\bf A}^{\intercal}{\bf \tilde{B}}_j\oplus_q {\bf \tilde{B}}_j^{\intercal}{\bf A}\nonumber\\
&=\begin{bmatrix}
d_{1j}\oplus_q d_{1j} & d_{1j}\oplus_q d_{2j} &  \hdots & d_{1j}\oplus_q d_{lj}\\
d_{2j}\oplus_q d_{1j} & d_{2j}\oplus_q d_{2j} &  \hdots & d_{2j}\oplus_q d_{lj}\\
 \vdots & \vdots & \ddots & \vdots \\
d_{lj}\oplus_q d_{1j} & d_{lj}\oplus_q d_{2j} &  \hdots & d_{lj}\oplus_q d_{lj}
\end{bmatrix}\ ,\nonumber
\end{align}
which is a symmetric matrix with $l$ unknowns and $\frac{l(l-1)}{2}\geq l$ linearly independent equations for $l\geq 2$ and $q> 2$. Hence, ${\bf D}_j$, for each $j=1,\dots,l$, as well as ${\bf D}$ can be recovered.

\subsection{Proof of Corollary~\ref{cor:general_matrix_q3}}
\label{App:cor:general_matrix_q3}

The sum rate for distributed encoding of $({\bf A},{\bf B})$ is given as
\begin{align}
\RSWsum=H({\bf A},{\bf B})&=H({\bf A}_1,{\bf B}_1,{\bf A}_2,{\bf B}_2)\nonumber\\
&=H({\bf A}_1,{\bf B}_1)\nonumber\\
&=m(h(2\epsilon)+(1-2\epsilon)+h(p)) \ .\nonumber
\end{align}

Exploiting Proposition~\ref{prop:KW_sum_rate_for_general_matrix_product} to compute ${\bf A}^{\intercal}{\bf B}$, we can achieve
\begin{align}
\RKMsum&=2H({\bf A}\oplus_3{\bf \tilde{B}}_1\, , \, {\bf A}^{\intercal}{\bf A}\oplus_3{\bf \tilde{B}}_1^{\intercal}{\bf \tilde{B}}_1)\nonumber\\ 
&=2m h\Big(2\Big(\frac{1}{2}-\epsilon\Big)(1-p)+2\epsilon(1-p),\nonumber\\
&\hspace{3cm}2\Big(\frac{1}{2}-\epsilon\Big)p+2\epsilon p\Big) \nonumber\\
&+2h({\bf A}^{\intercal}{\bf A}\oplus_3{\bf \tilde{B}}_1^{\intercal}{\bf \tilde{B}}_1\,\vert\,{\bf A}\oplus_3{\bf \tilde{B}}_1)\nonumber\\
&\leq 2m h\Big(2\Big(\frac{1}{2}-\epsilon\Big)(1-p)+2\epsilon(1-p),\nonumber\\
&\hspace{1.4cm} 2\Big(\frac{1}{2}-\epsilon\Big)p+2\epsilon p\Big)+2\log_2(3)\ ,\nonumber
\end{align}
where the last step follows from using that
\begin{align}
{\bf A}^{\intercal}{\bf A}\oplus_3{\bf \tilde{B}}_1^{\intercal}{\bf \tilde{B}}_1=
\begin{bmatrix}
{\bf A}_1^{\intercal}{\bf A}_1\oplus_3{\bf B}_1^{\intercal}{\bf B}_1 & {\bf A}_1^{\intercal}{\bf A}_2\oplus_3{\bf B}_1^{\intercal}{\bf B}_1 \\
{\bf A}_2^{\intercal}{\bf A}_1\oplus_3{\bf B}_1^{\intercal}{\bf B}_1 & {\bf A}_2^{\intercal}{\bf A}_1\oplus_3{\bf B}_1^{\intercal}{\bf B}_1
\end{bmatrix} \ ,\nonumber
\end{align}
and evaluating the conditional entropy as
\begin{align}
h({\bf A}^{\intercal}&{\bf A}\oplus_3{\bf \tilde{B}}_1^{\intercal}{\bf \tilde{B}}_1 \,\vert\,{\bf A}\oplus_3{\bf \tilde{B}}_1)\nonumber\\
&=h({\bf A}_1^{\intercal}{\bf A}_1\oplus_3{\bf B}_1^{\intercal}{\bf B}_1,{\bf A}_1^{\intercal}{\bf A}_2\oplus_3{\bf B}_1^{\intercal}{\bf B}_1,\nonumber\\
&\hspace{3.95cm}{\bf A}_2^{\intercal}{\bf A}_1\oplus_3{\bf B}_1^{\intercal}{\bf B}_1\,\vert\,{\bf A}\oplus_3{\bf \tilde{B}}_1)\nonumber\\
&=h(\big\{\sum\limits_{i=1}^m a_{ij}^2\oplus_3 b_{i1}^2\big\}_{j=1}^2,\sum\limits_{i=1}^m a_{i1}a_{i2}\oplus_3 b_{i1}^2\,\vert\,{\bf A}\oplus_3{\bf \tilde{B}}_1)\nonumber\\
&\overset{(a)}{=}h(\sum\limits_{i=1}^m a_{i1}^2\oplus_3 b_{i1}^2,\sum\limits_{i=1}^m a_{i2}^2\oplus_3 b_{i1}^2\,\vert\,{\bf A}\oplus_3{\bf \tilde{B}}_1)\nonumber\\
&\overset{(b)}{\leq} 2\log_2(3) \ ,\nonumber
\end{align}
where $(a)$ follows from that $a_{i1}a_{i2}\oplus_3 b_{i1}^2$ can be recovered given the side information ${\bf A}\oplus_3{\bf \tilde{B}}_1=\big\{a_{ij}\oplus_3 b_{i1}\big\}_{i,j}$ for $i=1,2,\dots,m$ and $j=1,2$, and given $\big\{a_{ij}^2\oplus_3 b_{i1}^2\big\}_j$, for $j=1,2$. More specifically, the receiver can recover
\begin{align}
\label{prod_var}
2a_{ij}b_{i1}=(a_{ij}\oplus_3 b_{i1})^2-(a_{ij}^2\oplus_3 b_{i1}^2)\ , \quad j=1,2 \ . 
\end{align}
Hence, using the side information ${\bf A}\oplus_3{\bf \tilde{B}}_1$ and (\ref{prod_var}), the receiver can recover $a_{i1}a_{i2}\oplus_3 b_{i1}^2$ as follows:
\begin{align}
a_{i1}a_{i2}\oplus_3 b_{i1}^2&=(a_{i1}\oplus_3 b_{i1})(a_{i2}\oplus_3 b_{i1})\nonumber\\
&-(a_{i1}b_{i1}\oplus_3 a_{i2}b_{i1}) \ . \nonumber
\end{align} 
Finally, step $(b)$ follows from exploiting that both variable $\sum\limits_{i=1}^m a_{i1}^2\oplus_3 b_{i1}^2$ and variable $\sum\limits_{i=1}^m a_{i2}^2\oplus_3 b_{i1}^2$ reside in $\mathbb{F}_3$.

\bibliographystyle{IEEEtran}
\bibliography{derya}

\end{document}